\documentclass[journal]{IEEEtran}
\usepackage{amsmath}
\usepackage{amsfonts}
\usepackage{amsthm}
\usepackage{caption}
\usepackage{graphicx}
\usepackage{subcaption}

\newtheorem{theorem}{Theorem}
\newtheorem{definition}{Definition}
\newtheorem{lemma}{Lemma}

\usepackage[utf8]{inputenc}
\usepackage{algorithm}

\begin{document}


\title{A Framework for Distributed and Compositional Stability Analysis of Power Grids}


\author{Stefanos Baros, Andrey Bernstein and Nikos Hatziargyriou
\thanks{This work was authored in part by the National Renewable Energy
Laboratory, operated by Alliance for Sustainable Energy, LLC, for the U.S.
Department of Energy (DOE) under Contract No. DE-AC36-08GO28308. Funding provided by DOE Office of Electricity, Advanced Grid
Modeling Program, through agreement NO. 33652.}
\thanks{S. Baros and A. Bernstein are with the Power Systems Engineering
Center, National Renewable Energy Laboratory, Golden, CO 80401, USA
stefanos.baros@nrel.gov, andrey.bernstein@nrel.gov.}
\thanks{Nikos Hatziargyriou is with the ECE Department of National Technical University of Athens, Greece, nh@power.ece.ntua.gr.}}

\maketitle

\begin{abstract}
    Operating modern power grids with stability guarantees is admittedly imperative. Classic stability methods are not well-suited for these dynamic systems as they involve centralized gathering of information and computation of the system's eigenvalues, processes which are oftentimes not privacy-preserving and computationally burdensome. System operators (SOs) would nowadays have to be able to quickly and efficiently assess small-signal stability as the power grid operating conditions change more dynamically while also respect the privacy of the distributed energy resources (DERs). Motivated by all these, in this paper we introduce a framework that comprises a \textit{computationally efficient,} \textit{privacy preserving, distributed} and \textit{compositional} stability assessment method. Our proposed method first calls for representative agents at various buses to exchange information with their neighbors and design their local controls in order to meet some local stability conditions. Following that, the agents are required to notify the SO whether their local conditions are satisfied or not. In case the agents cannot verify their local conditions they can augment their local controls using a global control input.  The SO can then warrant stability of the interconnected power grid by assembling the local stability guarantees, established by the agents, in a \textit{compositional} manner. We analytically derive the local stability conditions and prove that when they are collectively satisfied stability of the interconnected system ensues. We illustrate the effectiveness of our proposed DSA method via a numerical example centered around a three-bus power grid.
\end{abstract}
\section{Introduction}
The main constituents of bulk power grids have traditionally been bulk synchronous generators with large inertias and slow-varying loads. Stability of these systems has been well-maintained by the control systems of generators \cite{Nikos_stability}. Over the recent years though, the increasing penetration of renewable energy resources (RERs) and fast-varying demand response resources has resulted in significant reduction of the inertia of these systems and  rendered stability a more critical issue.  It is now imperative to develop advanced control methodologies that can enable  wind power plants \cite{Baros_distr_torque}, \cite{Baros_distr_wind_stor}, \cite{Baros_distr_consensus_wind_control} and other RERs \cite{Baros_emerging_control} to provide ancillary services in order to enhance power grid stability. At the same time, the electricity markets have to be restructured so at to promote optimal operation of power grids with renewables \cite{Baros_dyn_markets}, \cite{Baros_rel_contracts}, \cite{Baros_opt_AGC}. 
\par  The trend towards decentralizing  power generation \cite{Nikos_microgrids} has rendered stability assessment of modern power grids very challenging. This is particularly true for distribution grids and large microgrids as their operators, have now hard time obtaining accurate information pertaining to the  models of the numerous and largely heterogeneous distributed energy resources (DERs) \cite{Nikos_DER} scattered throughout vast geographical areas \cite{Towards_dist_stab} in order to carry out stability analysis. In addition, classical centralized methods for stability assessment, involving time-domain simulation \cite{Numerical_simulations}, eigenvalue analysis \cite{eig_analysis}, and direct methods \cite{Direct_stab_anal}, are inherently inefficient and not privacy preserving. Hence, they are not well-suited for power grids that comprise numerous distributed, complex and heterogeneous resources that place a lot of value on their privacy. These methods fall short primarily because they are computationally expensive and thus slow while they require extensive exchange of information through dedicated communication channels, part of which, is often private. With all these in mind, one might ask - \textit{How can we evaluate and certify small-signal stability of modern power grids in a fully distributed, computationally efficient and privacy-preserving way}? In this paper, we aim to tackle this particular question. The scientific community has recently started exploring new distributed methodologies for assessing stability of power grids which respect the privacy of energy resources and are computationally efficient. Several publications that we singled out as representative of this line of work are  \cite{Towards_dist_stab}, \cite{Hill_distr_stab}, \cite{LeXie_online_dynamic_security}, \cite{Soumya_sos_approach}, \cite{Soumya_distr_contr} and  \cite{Tabuada}.
\par In \cite{Towards_dist_stab},  a passivity-based approach for establishing system-wide stability is proposed. This approach is quite conservative as it relies on the rather strict condition that the local bus dynamics have to be passive. In \cite{Hill_distr_stab}, a method that enables a bus to reconstruct the system dynamic Jacobian matrix from data and use that to carry out stability evaluation locally is presented. This method requires each bus to perform a lot of computations locally and exchange a significant amount of information with its neighbors. In \cite{LeXie_online_dynamic_security} and \cite{LeXie_interactive_control},  approaches based on dissipative system theory are proposed that can enable microgrid agents to assess system-wide stability in a distributed fashion. These approaches are not fully distributed as they still require the system operator to compute the Jacobian of the system and communicate information to multiple agents. In \cite{Soumya_sos_approach} and \cite{Soumya_distr_contr},  sum-of-squares approaches are introduced for distributed stability assessment and control of large-scale nonlinear systems using vector Lyapunov functions. These methods, although very effective for nonlinear systems with hard-to-compute Lyapunov functions, are not scalable for large-scale systems. In \cite{Turitsyn_microgrids}, the authors derived distributed stability criteria that are particularly tailored to droop-controlled inverter-based microgrids. Finally, in \cite{Tabuada} and \cite{LeXie_transient_stab}, approaches for compositional transient stability analysis are proposed, which is not the focus of this paper. 
\par \textbf{\textit{Contributions}}. In this paper, we focus on the problem of assessing and certifying small-signal stability of a power grid in a \textit{fully distributed}, \textit{computationally efficient} and \textit{privacy-preserving way}. Our distributed stability assessment (DSA) methodology is general enough and it can be applied to microgrids, transmission and distribution grids. Our main contributions are highlighted as follows.
\begin{itemize}
\item We build on several results from \cite{Siljak_dec}, \cite{Siljak_large_scale} and analytically derive a simple condition that can be assessed locally by bus agents, leading to guaranteed small-signal stability of a power grid. We also derive a relaxed version of this distributed stability condition.
\item We design an algorithm for DSA by leveraging the derived local condition. The proposed algorithm requires multiple agents to exchange limited information with their neighbors, design the bus-level control laws in order to meet the local stability condition and then inform the system operator whether the underlying condition is satisfied. The SO can then establish stability of the overall system by combining the local stability guarantees and invoking our main theorem.
\item We introduce a global control design approach that can enable the agents to meet their local stability condition, in case using local controls only is ineffective, and thus enable stabilization of the overall system.
\item We numerically show that the proposed distributed methodology can result in a small-signal stable overall system via a detailed three-bus power grid example.
\end{itemize}
\par The rest of the paper is structured as follows. In Section II, we introduce our distributed stability assessment (DSA) methodology and the main results of the paper. In Section III, we illustrate the effectiveness and practicality of our methodology via a detailed three-bus power grid example. In Section IV,  we conclude the paper with some remarks and an account of future work.
\section{Framework for Distributed and Compositional Stability Analysis}
In this section, we develop our methodology for distributed stability assessment of power grids in a step-by-step fashion.
\subsection{Stability of Decoupled Subsystems}
We depart from the dynamical model of a power system arising from the interconnection of $N$ linear subsystems
\begin{equation}
\Sigma: \dot{x}_i=A_i x_i+\sum_{j\in\mathcal{N}_i} A_{ij}x_j,\hspace{3mm} i\in\mathcal{N} \label{interconnected_sys}
\end{equation}
In this representation, $x_i\in\mathbb{R}^{n_i}$ denotes the state-vector corresponding to subsystem $i$, $\mathcal{N}:=\{1,...,N\}$ the set of all subsystems and $\mathcal{N}_i$ the set of subsystems that are adjacent and directly connected to subsystem $i$. The closed-loop dynamics of each isolated decoupled system can be described by
\begin{align}
 \dot{x}_i=A_ix_i,\hspace{3mm} i\in\mathcal{N} \label{isolated_sys}
\end{align}
We assume that $A_i$ has full rank  so that the equilibrium of each decoupled subsystem is $x_{i}^{\star}=0_{n_i}$ and that the local controllers are tuned so that this is asymptotically stable. A certificate of this stability property is a Lyapunov function
\begin{align}
V_i(x_i)=x_i^\top P_i x_i,\; V_i\in\mathbb{R}_{+}
\end{align}
for which it holds
\begin{align}
&P_i\succ 0\\
&A_i^\top P_i+P_iA_i=-Q_i,\;\;Q_i\succ 0
\end{align}
These conditions translate into the following inequalities
\begin{align}
V_i(x_i)&>0, \;\;\forall x_i\in\mathbb{R}^{n_i}\setminus 0_{n_i} \\
\dot{V}_i(x_i)&<0, \;\;\forall x_i\in\mathbb{R}^{n_i}\setminus 0_{n_i} 
\end{align}
where $\dot{V}_i(x_i)$ is the derivative of the Lyapunov function $V_i$ along the trajectories of the decoupled system \eqref{isolated_sys}. By exploiting the stability certificates for the decoupled subsystems, we now seek to derive criteria that would allow us to establish stability of the interconnected system $\Sigma$ in \eqref{interconnected_sys} in a fully distributed and compositional manner.
\subsection{Distributed Stability Condition}
We pose the main question in this paper as follows: \textit{Granted that all the isolated subsystems $i\in\mathcal{N}$ described by \eqref{isolated_sys} are asymptotically stable, is there an additional distributed condition, that can be assessed locally by each subsystem $i\in\mathcal{N}$, than can result in stability of the interconnected system $\Sigma$ in \eqref{interconnected_sys}?} 
\par We address this question following a constructive approach; we derive a distributed condition that warrants stability of the interconnected system. Our analysis builds on several key results from \cite{Siljak_dec}, \cite{Siljak_large_scale}. Let us first consider a candidate Lyapunov function for the interconnected system $V(x)$ constructed as:
\begin{align}
    V(x)=\sum_{i\in\mathcal{N}}V_i(x_i),\hspace{3mm} V(x)>0,\hspace{3mm} \forall x\in\mathbb{R}^{\overline{n}}\setminus 0_{\overline{n}} 
\end{align}
where $\overline{n}=\sum_{i\in\mathcal{N}}n_i$. Clearly, our goal now is to derive conditions under which $\dot{V}(x)<0,\;\;\forall x\in\mathbb{R}^{\overline{n}}\setminus 0_{\overline{n}}  $. To carry out that, we compute the derivative of $V(x)$ along the trajectories of the interconnected system \eqref{interconnected_sys} as follows
\begin{align}
    \dot{V}(x)=\sum_{i\in\mathcal{N}}\dot{V}_i(x_i)=\sum_{i\in\mathcal{N}}(\nabla V_i^\top A_ix_i+\nabla V_i^\top \sum_{j\in\mathcal{N}_i} A_{ij}x_j).
\end{align}
We know that
\begin{align}
\nabla V_i^\top&=x_i^\top(P_i+P_i^\top)\\
x_i^\top P_i^\top A_i x_i&=x_i^\top A_i^\top P_i x_i
\end{align}
Given these, one readily obtains
\begin{align}
    \sum_{i\in\mathcal{N}}\nabla V_i^\top  A_ix_i=\sum_{i\in\mathcal{N}}x_i^\top(A_i^\top P_i+P_iA_i)x_i=-\sum_{i\in\mathcal{N}}x_i^\top Q_ix_i.
\end{align}
In light of that,  the following inequality arises
\begin{align}
    \dot{V}(x)&=\sum_{i\in\mathcal{N}}\dot{V}_i(x_i)\leq-\sum_{i\in\mathcal{N}}\lambda_{min}(Q_i)\|x_i\|_2^2\nonumber\\
    &+\sum_{i\in\mathcal{N}}\|\nabla V_i^\top\|_2\sum_{j\in\mathcal{N}_i} \|A_{ij}x_j\|_2. \label{Vdot_ineq}
\end{align}
Further, we know that
\begin{align}
    \|A_{ij}x_j\|_2^2&=x_j^\top A_{ij}^\top A_{ij}x_j\leq \lambda_{max}(A_{ij}^\top A_{ij})\|x_j\|_2^2\\
    	\Rightarrow \|A_{ij}x_j\|_2&\leq \sqrt{\lambda_{max}(A_{ij}^\top A_{ij})}\|x_j\|_2. \label{ineq_inter}
\end{align}
We assume that $P_i=P_i^\top$, and derive a bound on $\|\nabla V_i^\top\|_2$ as
\begin{align}
\|\nabla V_i^\top\|_2=2\| x_i^\top P_i\|_2\leq 2 \|x_i\|_2 \|P_i\|_2 =  2 \sigma_{max}(P_i)\|x_i\|_2. \label{ineq_grad1}
\end{align}
For square symmetric matrices. we know that
\begin{align}
\sigma_{max}(P_i)&=\sqrt{\lambda_{max}(P_i^\top P_i)}=\lambda_{max}(P_i) \label{sigma_max}
\end{align}
where $\lambda(\cdot)$ is the eigenvalue operator. Taking into account \eqref{sigma_max},  inequality \eqref{ineq_grad1} can be finally written as
\begin{align}
    \|\nabla V_i^\top\|_2 \leq  2 \lambda_{max}(P_i)\|x_i\|_2 \label{ineq_grad2}
\end{align}
In light of inequalities \eqref{ineq_inter} and \eqref{ineq_grad2}. we can express \eqref{Vdot_ineq} as
\begin{align}
    \dot{V}(x)&=\sum_{i\in\mathcal{N}}\dot{V}_i(x_i)\leq-\sum_{i\in\mathcal{N}}\lambda_{min}(Q_i)\|x_i\|_2^2\nonumber\\
    &+\sum_{i\in\mathcal{N}}2 \lambda_{max}(P_i)\|x_i\|_2\sum_{j\in\mathcal{N}_i} \sqrt{\lambda_{max}(A_{ij}^\top A_{ij})}\|x_j\|_2. \end{align}
    We can express this inequality in matrix form as
    \begin{align}
    \dot{V}(x)\leq-\frac{1}{2}\phi(x)^\top (S+S^\top)\phi(x) \end{align}
       where
    \begin{align}
    \phi(x):=\Big[\|x_{1}\|_2,...,\|x_{i}\|_2,...,\|x_{N}\|_2 \Big]^\top \in\mathbb{R}^{N}.
    \end{align}
    The matrix $S$ is defined as
    \begin{align}
    S:=\begin{cases} s_{ii}=\lambda_{min}(Q_i), & i=j \\
      s_{ij}=-2\lambda_{max}(P_i) \sqrt{\lambda_{max}(A_{ij}^\top A_{ij})}, & i\neq j,\;\; j\in\mathcal{N}_i\\
      s_{ij}=0, & \text{otherwise}\end{cases}
    \end{align}
 The matrix $(S+S^\top)$ is positive definite if and only if $S$ is an M-matrix. The definition of an M-matrix is given next.
    \begin{definition}
    An M-matrix is a matrix that has off-diagonal entries less than or equal to zero and eigenvalues whose real parts are nonnegative.\label{def_M_matrix}
    \end{definition}
    Given this definition, the following lemma ensues.
    \begin{lemma}[\cite{Khalil}]
    Diagonally dominant matrices with off-diagonal entries less than or equal to zero are M-matrices.\label{lemma_M_matrix}
    \end{lemma}
     Next, we restate a Theorem from \cite{Siljak_dec}, \cite{Siljak_large_scale} that is useful in establishing stability of the interconnected system \eqref{interconnected_sys}.
    \begin{theorem}[\cite{Siljak_dec}, \cite{Siljak_large_scale}]
    The interconnected system $\Sigma$ is asymptotically stable if $S$ is an M-matrix. \label{Theorem1}
    \end{theorem}
     In light of definition~\ref{def_M_matrix} and Lemma~\ref{lemma_M_matrix}, we now state a theorem that comprises a distributed condition, that can be examined locally by each subsystem, resulting in stability of the interconnected system $\Sigma$.
    \begin{theorem}
    The interconnected system $\Sigma$ is asymptotically stable if 
    \begin{align}
       &|\lambda_{min}(Q_i)|>\sum_{j\in\mathcal{N}_i}|-2\lambda_{max}(P_i) \sqrt{\lambda_{max}(A_{ij}^\top A_{ij})}|,\hspace{3mm} \forall i\in\mathcal{N} \label{condition1}
    \end{align}
holds in addition to $A_i^\top P_i+P_iA_i=-Q_i$ where $P_i, Q_i\succ 0$ $\forall i\in\mathcal{N}$. 
\label{Theorem_dist_cond}
    \end{theorem}
    \begin{proof}
    Follows by combining Theorem \ref{Theorem1} and Lemma~\ref{lemma_M_matrix}.
    \end{proof}
    Representative subsystem agents can probe whether the distributed stability condition of Theorem~\ref{Theorem_dist_cond} is satisfied by only having information about $P_i, Q_i$ and the interconnection matrices $A_{ij}$. In the case the agents cannot make sure that the distributed stability condition holds, they have two options. They can redesign their local controls, and through that manipulate the eigenvalues of the matrices $A_i$, and/or, implement a global control input, to minimize the effect of the interconnections $A_{ij}$, until the condition is finally satisfied. The system operator can then leverage the local stability guarantees to establish stability of the overall system. Hence, ensuring stability of the overall system is delegated to multiple agents with whom the  system operator communicates frequently.  It is important to note though that, although the distributed stability condition offers a computationally efficient way to establish stability of the overall system, it is only sufficient and not necessary.

\subsection{Relaxed Distributed Stability Condition}
The condition postulated in Theorem~\ref{Theorem_dist_cond} is only sufficient and can often be quite conservative. Realizing that, one can attempt to relax its conservativeness by choosing matrices $P_i$ and $Q_i$ that solve the following problem
\begin{align}
    \max_{P_i, Q_i} \hspace{3mm}&\frac{\lambda_{min}(Q_i)}{\lambda_{max}(P_i)}\nonumber\\
    \text{such that}\hspace{3mm}&A_i^\top P_i+P_iA_i=-Q_i.
\end{align}
In other words, choose $P_i, Q_i$ that lead to maximization of the ratio $\lambda_{min}(Q_i)/\lambda_{max}(P_i)$ under the restriction that $A_i^\top P_i+P_iA_i=-Q_i$. To solve this problem according to \cite{Siljak_dec}, \cite{Siljak_large_scale}, we can choose a non-singular matrix $T_i$ such that $\Lambda_i=T_i^{-1}A_iT_i$ is semi-simple, i.e., its diagonal is composed by, first of all the complex eigenvalues and then, all the real eigenvalues of $A_i$. Under such a transformation,  the interconnected system \eqref{interconnected_sys} can be recast to
\begin{equation}
\Sigma_{trans}: \dot{\tilde{x}}_i=\Lambda_i \tilde{x}_i+\sum_{j\in\mathcal{N}_i} \tilde{A}_{ij}\tilde{x}_j,\hspace{3mm} i\in\mathcal{N} \label{interconnected_sys_transf}
\end{equation}
where $\tilde{A}_{ij}=T_{i}^{-1}A_{ij}T_j$ and $x_i=T_i\tilde{x}_i$. 
\par By choosing $\tilde{P}_i=\theta \mathbb{I}_{n_i}$ where, $\theta$ is a positive constant and $\mathbb{I}_{n_i}$ the identity matrix, we obtain $\tilde{Q}_i=-2\theta\Lambda_i=2\theta \text{diag}\{\sigma_1^i,...,\sigma_{n_i}^i\}$ from $\Lambda_i^\top \tilde{P}_i+\tilde{P}_i\Lambda_i=-\tilde{Q}_i$. These choices result in the maximum ratio
\begin{align}
    \frac{\lambda_{min}(\tilde{Q}_i)}{\lambda_{max}(\tilde{P}_i)}=\frac{2\theta\min\{\sigma_1^{i},...,\sigma_{n_i}^{i}\}}{\theta}=2\sigma_{M}^i
\end{align}
where $-\sigma_M^i$ is the real part of the maximum eigenvalue of $A_i$. By considering $V(\tilde{x})=\sum_{i\in\mathcal{N}}\tilde{x}_i^\top \tilde{P}_i \tilde{x}_i$ we arrive at
\begin{align}
    \dot{V}(\tilde{x})&=\sum_{i\in\mathcal{N}}\dot{V}_i(\tilde{x}_i)\leq-\sum_{i\in\mathcal{N}}\lambda_{min}(\tilde{Q}_i)\|\tilde{x}_i\|_2^2\nonumber\\
    &+\sum_{i\in\mathcal{N}}\|\nabla V_i^\top\|_2\sum_{j\in\mathcal{N}_i} \|\tilde{A}_{ij}\tilde{x}_j\|_2. \label{Vdot_ineq2}
\end{align}
Given that
\begin{align}
    	 \|\tilde{A}_{ij}\tilde{x}_j\|_2&\leq \sqrt{\lambda_{max}(\tilde{A}_{ij}^\top \tilde{A}_{ij})}\|\tilde{x}_j\|_2 \label{ineq_inter2}
\end{align}
and that $\|\nabla V_i^\top\|_2$ is bounded as
\begin{align}
\|\nabla V_i^\top\|_2& \leq  2 \sigma_{max}(\tilde{P}_i)\|\tilde{x}_i\|_2 \leq 2\theta\|\tilde{x}_i\|_2 \label{ineq_grad_theta}
\end{align}
we can express \eqref{Vdot_ineq2} as
\begin{align}
    \dot{V}(x)&\leq-\sum_{i\in\mathcal{N}}2\theta \min\{\sigma_1^{i},...,\sigma_{n_i}^{i}\}\|\tilde{x}_i\|_2^2\nonumber\\
    &+\sum_{i\in\mathcal{N}}2 \theta \|\tilde{x}_i\|_2\sum_{j\in\mathcal{N}_i} \sqrt{\lambda_{max}(\tilde{A}_{ij}^\top \tilde{A}_{ij})}\|\tilde{x}_j\|_2. \end{align}
This inequality can be expressed in the matrix form
    \begin{align}
\dot{} \dot{V}(\tilde{x})\leq-\theta\tilde{\phi}(\tilde{x})^\top (\tilde{S}+\tilde{S}^\top)\tilde{\phi}(\tilde{x}) \end{align}
where 
 \begin{align}
    \tilde{\phi}(\tilde{x}):=\Big[\|\tilde{x}_{1}\|_2,...,\|\tilde{x}_{i}\|_2,...,\|\tilde{x}_{N}\|_2 \Big]^\top \in\mathbb{R}^{N}
    \end{align}
and $\tilde{S}$ is a new matrix defined as
 \begin{align}
    \tilde{S}=[\tilde{s}_{ij}]:=\begin{cases} \tilde{s}_{ii}=\sigma_M^i, & i=j \\
      \tilde{s}_{ij}=- \sqrt{\lambda_{max}(\tilde{A}_{ij}^\top \tilde{A}_{ij})},& i\neq j,\;\; j\in\mathcal{N}_i\\
       \tilde{s}_{ij}=0,& \text{otherwise}\end{cases}
    \end{align}
    The following theorem from \cite{Siljak_dec}, \cite{Siljak_large_scale} can be used to establish stability of the \textit{transformed} interconnected system \eqref{interconnected_sys_transf}.
    \begin{theorem}[\cite{Siljak_dec}, \cite{Siljak_large_scale}]
    The interconnected system $\Sigma_{trans}$ is asymptotically stable if $\tilde{S}$ is an M-matrix.\label{Theorem3}
    \end{theorem}
   Next, we set forth a \textit{relaxed} distributed condition that translates into stability of the interconnected system $\Sigma_{trans}$.     \begin{theorem}
    The interconnected system $\Sigma_{trans}$ is asymptotically stable if 
    \begin{align}
       &|\sigma_M^i |>\sum_{j\in\mathcal{N}_i}|-\sqrt{\lambda_{max}(\tilde{A}_{ij}^\top \tilde{A}_{ij})}|,\hspace{3mm} \forall i\in\mathcal{N} \label{condition3}
    \end{align}
holds and $\Lambda_i^\top \tilde{P}_i+\tilde{P}_i\Lambda_i=-\tilde{Q}_i$ where $\tilde{P}_i, \tilde{Q}_i\succ 0$ $\forall i\in\mathcal{N}$.
\label{Theorem_dist_cond_transf}
    \end{theorem}
      \begin{proof}
    Follows by combining Theorem \ref{Theorem3} and Lemma~\ref{lemma_M_matrix}.
    \end{proof}
   By examining Theorem~\ref{Theorem_dist_cond_transf}, one can easily notice that Condition~\eqref{condition3} is less conservative than Condition~\eqref{condition1}  as it enforces a more relaxed bound on the interconnection matrices. There might be cases though where the \textit{transformed} interconnection matrices $\tilde{A}_{ij}$ turn out to be greater than the original matrices $A_{ij}$, overcoming the gain in the size of the diagonal elements $\tilde{s}_{ii}$. In such cases, stability of the interconnected system better be studied using the original state-space.

\subsection{Stabilization}
The conditions stated in Theorems~\ref{Theorem_dist_cond} and \ref{Theorem_dist_cond_transf} are formulated with regard to the closed-loop systems \eqref{interconnected_sys} and \eqref{interconnected_sys_transf}. We now elaborate on the various ways the agents can design their controls to increase their chances of satisfying these conditions.  We depart from the following system form
\begin{align}
\dot{x}_i=\hat{A}_ix_i+\sum_{j\in\mathcal{N}_i}\hat{A}_{ij}x_j+B_iu_i \label{statespace_with_input}   
\end{align}
We consider the control input $u_i\in\mathbb{R}^{m_i}$
\begin{align}
    u_i&=u_{i}^l+u_{i}^g \label{controlinput2}
    \end{align}
constructed by the superposition of two control inputs
\begin{align}
    \text{(local input):\;} u_{i}^l&=-K_i^\top x_i,\;\;\;&&u_{i}^l\in\mathbb{R}^{m_i}\\
  \text{(global input):\;}  u_i^g&=-\sum_{j\in\mathcal{N}_i}K_{ij}^\top x_j,\;\;\;&&u_{i}^g\in\mathbb{R}^{m_i}
\end{align}
The local and global control inputs $u_i^l$ and $u_i^g$ can be properly designed to result in stability of the local decoupled subsystems and  minimization of their interconnection terms, respectively.  By assuming that each local subsystem is fully controllable with respect to the input $u_i$ and closing the loop via the control input \eqref{controlinput2}, we arrive at:
\begin{align}
\dot{x}_i=A_ix_i+\sum_{j\in\mathcal{N}_i}A_{ij}x_j
\end{align}
where $A_i=(\hat{A}_i-B_iK_i^\top)$ and $A_{ij}=(\hat{A}_{ij}-B_iK_{ij}^\top)$ are the closed-loop system matrices. As previously discussed, when one aims to design controls to meet Condition~\eqref{condition3}, he better use the transformed state-space representation. By employing a transformation matrix $T_i$, one can recast the system \eqref{statespace_with_input} to:
\begin{align}
\dot{\tilde{x}}_i=\tilde{A}_i^{'}\tilde{x}_i+\sum_{j\in\mathcal{N}_i}\tilde{A}^{'}_{ij}\tilde{x}_j+\tilde{B}_iu_i    
\end{align}
where
\begin{align}
    \tilde{A}_i^{'}&=T_i^{-1}\hat{A}_iT_i, \;\;\;   \tilde{A}_{ij}^{'}=T_i^{-1}\hat{A}_{ij}T_j,  \;\;\;\tilde{B}_i&=T_i^{-1}B_i
\end{align}
We consider the input $u_i$ as
\begin{align}
    u_i&=u_{i}^l+u_{i}^g \label{controlinput_trans2}\\
    u_{i}^l&=-\tilde{K}_i^\top \tilde{x}_i=-K_i^\top x_i\\
    u_i^g&=-\sum_{j\in\mathcal{N}_i}\tilde{K}_{ij}^\top \tilde{x}_j=-\sum_{j\in\mathcal{N}_i}K_{ij}^\top x_j
\end{align}
where the new control gains, $\tilde{K}_i, \tilde{K}_{ij}$, are associated with the original control gains, $K_i, K_{ij}$, through the relations
\begin{align}
    \tilde{K}_{i}^\top=K_{i}^\top T_i,\;\;\;\tilde{K}_{ij}^\top=K_{ij}^\top T_j
\end{align}
By closing the loop with the input $u_i$ given by \eqref{controlinput_trans2}, we end up with the \textit{closed-loop transformed} state-space model \eqref{interconnected_sys_transf}:
\begin{align}
\dot{\tilde{x}}_i=\Lambda_i \tilde{x}_i+\sum_{j\in\mathcal{N}_i}\tilde{A}_{ij}\tilde{x}_j \label{transformed_closed_loop}
\end{align}
where $\Lambda_i=\tilde{A}_i^{'}-\tilde{B}_i\tilde{K}_i^\top$ and $\tilde{A}_{ij}=\tilde{A}_{ij}^{'}-\tilde{B}_i\tilde{K}_{ij}^\top$. Standard pole placement techniques can be employed to find gains $\tilde{K}_i$ of the local input $u_i^l$ that result in \textit{Hurwitz} matrices $A_i$ and $\Lambda_i$. On the other hand,  the optimal control gains $\tilde{K}_{ij}$ of the global control input $u_i^g$ that yield minimization of interconnection terms i.e., gains that solve $\min_{\tilde{K}_{ij}} \|\tilde{A}_{ij}^{'}-\tilde{B}_i\tilde{K}_{ij}^\top\|$, can be  analytically computed as \cite{Siljak_dec}, \cite{Siljak_large_scale}
\vspace{-2mm}
\begin{align}
    \tilde{K}_{ij}^*=[(\tilde{B}_i^\top \tilde{B}_i)^{-1}\tilde{B}_i^\top \tilde{A}^{'}_{ij}]^\top \label{optimalglobalcontrol}
\end{align}
where $(\tilde{B}_i^\top \tilde{B}_i)^{-1}\tilde{B}_i^\top $ is the \textit{Moore-Penrose inverse} of $\tilde{B}_i$.
\par In the case  the agents cannot ensure that the distributed Condition \eqref{condition3} is satisfied, using only their local feedback control inputs $u_{i}^l$, they can resort to the global control inputs $u_i^g$. By exploiting these, they can minimize the effect of the interconnections, and that way increase their chances of satisfying the Condition \eqref{condition3}, at the expense of increased real-time information exchange. This is because implementation of the global control input $u_i^g$ mandates that the agents exchange with their neighbors, besides their transformation matrices $T_i$ and information about $A_{ij}$, their full state-space vector $x_i$.

   \subsection{Distributed Stability Analysis of  Power Grids}
 We now introduce our fully distributed algorithm for assessing small-signal stability of a power grid. \newline
\begin{algorithm}[H]
 \caption{Distributed Stability Assessment (DSA)}
\begin{enumerate}
 \item Each agent $i\in\mathcal{N}$ designs its local control input $u_i=u_i^{l}$ using the state-space representation \eqref{interconnected_sys} and computes the matrix $T_i$ that diagonalizes its local system matrix $A_i$ i.e., computes $T_i$ so that $\Lambda_i=T_i^{-1}A_iT_i$ holds.
      \item Each agent $i\in\mathcal{N}$ shares $T_i$ and information about interconnection term $A_{ij}$ with its neighbors $j\in\mathcal{N}_i$.
    \item Each agent $i\in\mathcal{N}$ uses measurements and information from its neighbors $j\in\mathcal{N}_i$ to reconstruct $\tilde{A}_{ij}$.
     
    \item Each agent $i\in\mathcal{N}$ checks whether \textit{Condition}~\eqref{condition3}  of \textit{Theorem}~\ref{Theorem_dist_cond_transf} is satisfied.
    \begin{itemize}
        \item     If the condition \textit{is met}, the agent $i$ broadcasts \textit{``Condition met"} to the system operator.
        \item If the conditions \textit{is not met}, the agent $i$ repeats the process from Step 1.
        \item After \textit{several unsuccessful} attempts of adjusting the local controls in order to satisfy \textit{Condition}~\eqref{condition3}, the agents can augment their control inputs with the global control inputs i.e., employ $u_i=u_i^{l}+u_i^{g}$,  and repeat Step 4.
    \end{itemize}
    \item Once the system operator receives the message \textit{``Condition met"} from all agents $i\in\mathcal{N}$ it notifies them that the system is certifiably stable and that the designed \textit{local} and \textit{global} controls can be implemented. 
\item Each agent $i\in\mathcal{N}$ who employs a global control $u_i^g$ in addition to a local control $u_i^l$, requests $x_j$ from its neighbors $j\in\mathcal{N}_i$  in real-time.
\end{enumerate}
\end{algorithm}
In the case the agents are unable to compute $\tilde{A}_{ij}$ because of lack of information, they can use worst-case bounds to carry out the stability assessment. The main steps involved in the implementation of Algorithm 1 are illustrated in Fig. \ref{dsa_algorithm_1} and \ref{dsa_algorithm_2}.

\begin{figure}
\begin{center}
\includegraphics[width=0.51\textwidth]{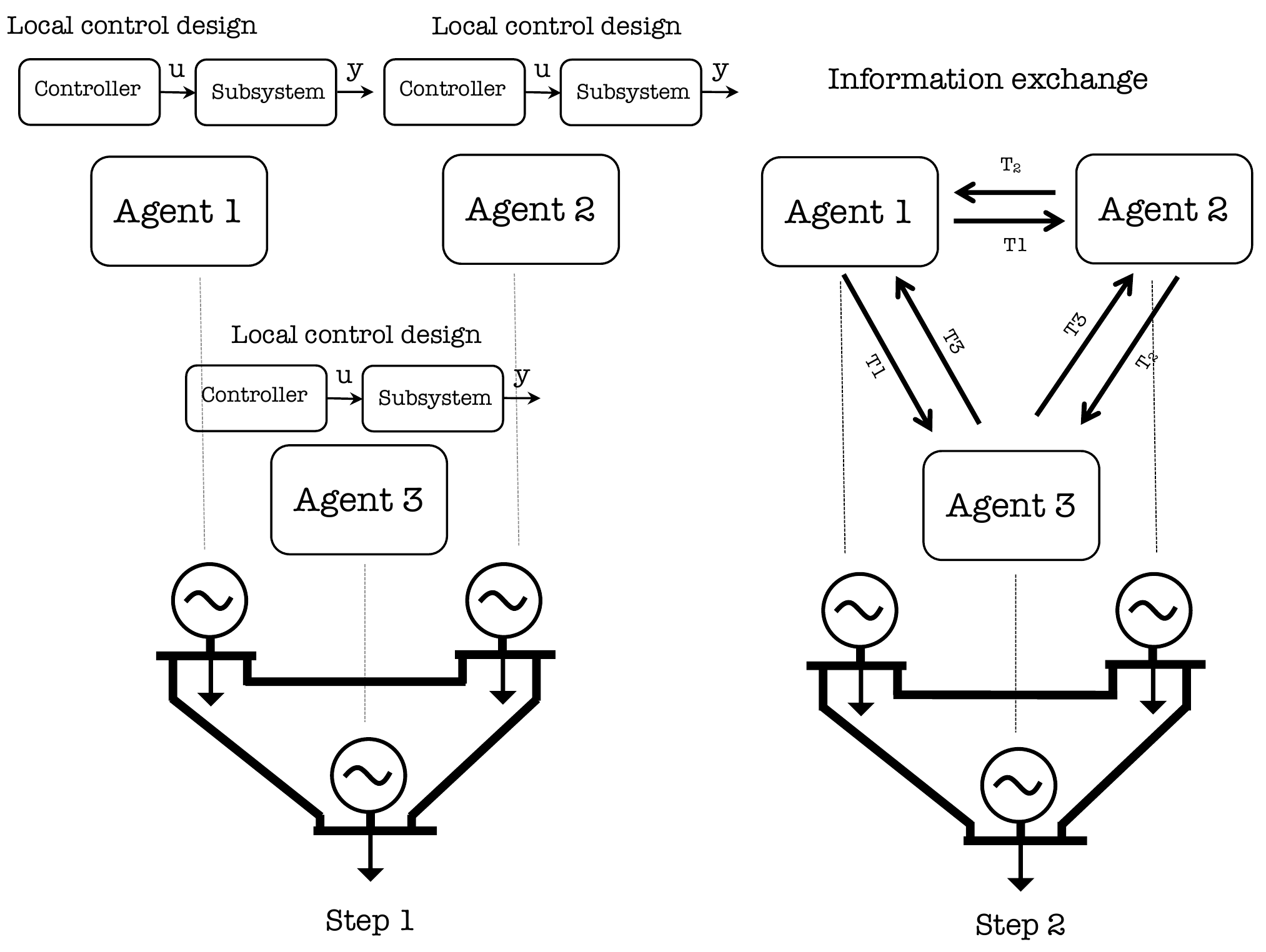}
\caption{Implementation of DSA algorithm, Steps 1 and 2.}
\label{dsa_algorithm_1}
\end{center}
\end{figure}
\begin{figure}
\begin{center}
\includegraphics[width=0.51\textwidth]{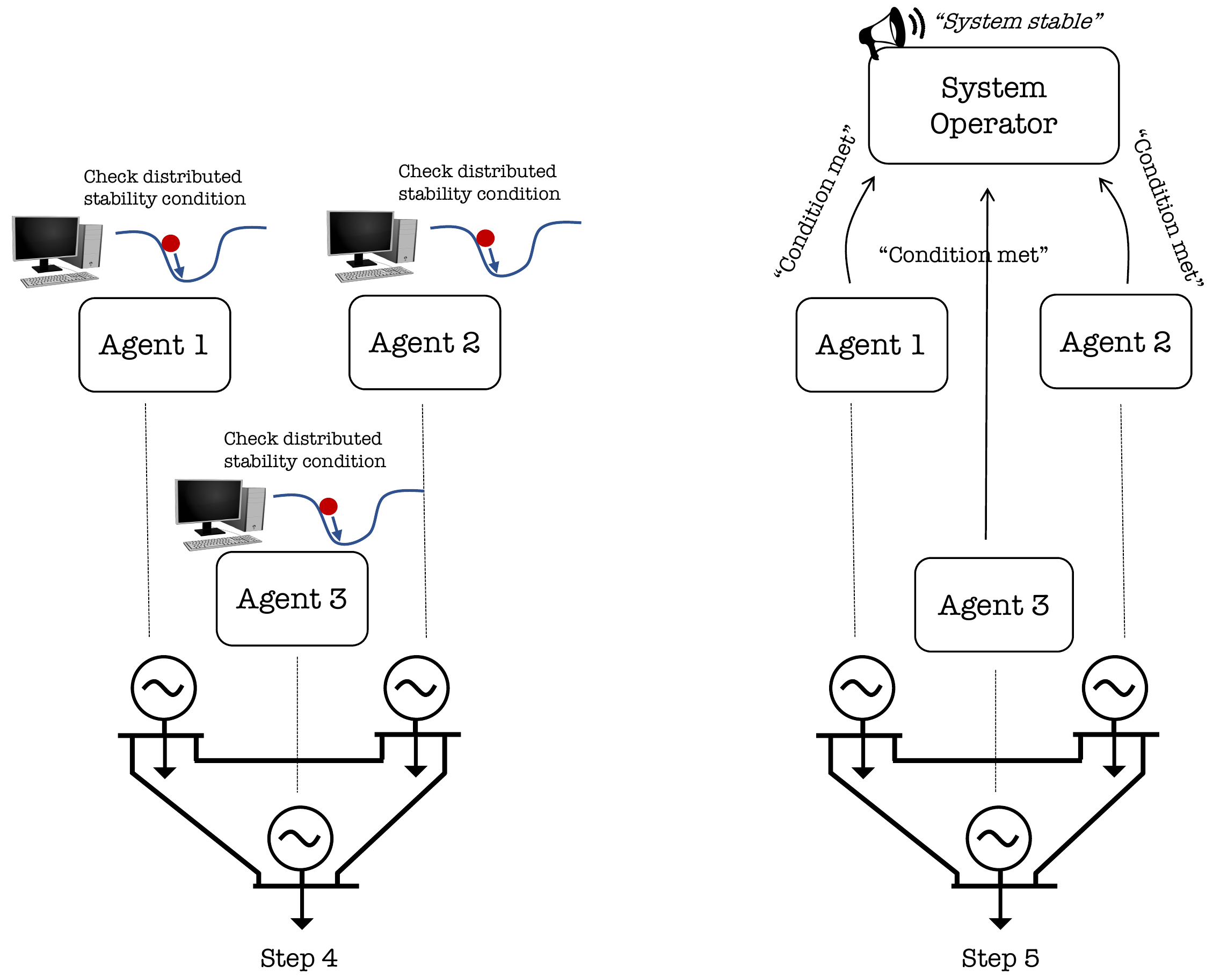}
\caption{Implementation of DSA algorithm, Steps 4 and 5.}
\label{dsa_algorithm_2}
\end{center}
\end{figure}
\begin{figure}
\begin{center}
\includegraphics[width=0.30\textwidth]{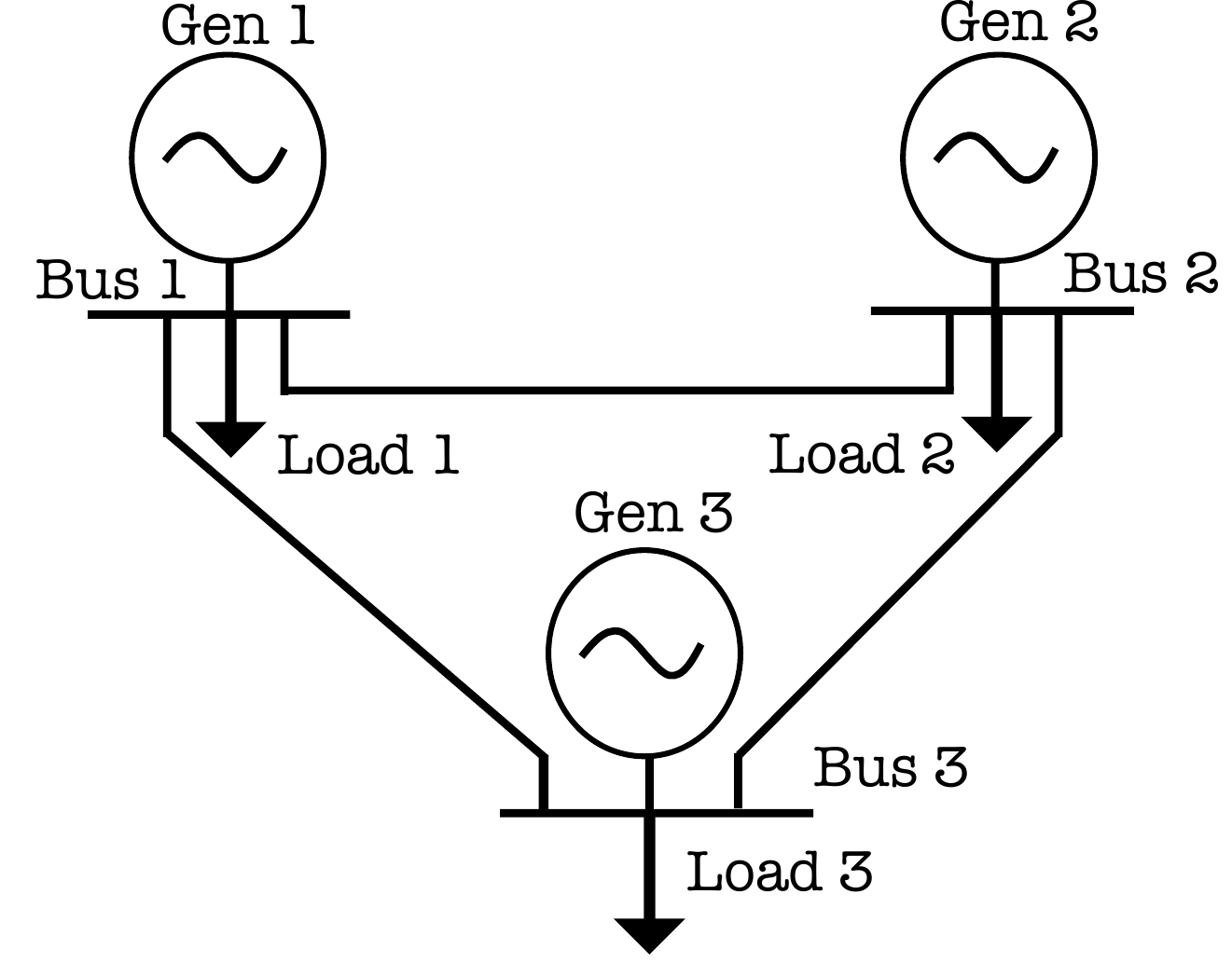}
\caption{Three-bus power grid.}
\label{three_bus_sys}
\end{center}
\end{figure}

\section{Illustrative Example}
In this section we focus on practical implementation and illustrate how a system operator can employ our proposed DSA algorithm to certify stability of an interconnected power grid by combining the \textit{local} stability certificates established by various agents. To carry out that, we leverage a detailed analysis on the three-bus power grid shown in Fig.~\ref{three_bus_sys}. We treat loads as fixed components and generators as dynamic ones. We consider a first-order turbine model and represent each generator via the interwined, linearized swing and turbine dynamics, expressed in \textit{per-unit linear} state-space form, as
\begin{align}
\Delta \dot{\delta}_i&=\Delta \omega_i \label{genmodel1}\\
\Delta \dot{\omega}_i&=\frac{\omega_b}{M_i}\Big(-\frac{D_i}{\omega_b}\Delta \omega_i +\Delta P_{m,i}-\Delta P_{L,i}\nonumber \\
&- \sum_{j\in\mathcal{N}_i}\frac{1}{X_{ij}}(\Delta \delta_i -\Delta\delta_j)\Big)\\
\Delta \dot{P}_{m,i}&=-\frac{\Delta P_{m,i}}{T_{T,i}}+\frac{u_{i}}{T_{T,i}}\label{genmodel3}
\end{align}
where $\omega_b=2\pi 60 \;rad/s$. The generator's control input $u_i\in\mathbb{R}$ does not follow the typical droop control law but is a free variable that the agents can adjust in the stabilization process. The state-variables of the generators $\Delta \delta_i, \Delta \omega_i,\Delta P_{m,i}\in\mathbb{R}$ are  defined as the deviations of the voltage angle, rotor speed and mechanical power from their equilibrium values. The state $\Delta \delta_i$ is given in radians, $\Delta \omega_i$ in rad/s and $\Delta P_{m,i}$ in per unit. The constant $\Delta P_{L,i}$ represents the load variation given also in per unit values. The terms $M_i, D_i, T_{T,i}, X_{ij}$ represent the inertia constant given in seconds, damping ratio in per-unit, turbine time-constant in seconds and line reactances in per unit, respectively. The equations \eqref{genmodel1}-\eqref{genmodel3} can be expressed in the standard state-space form
\begin{align}
\dot{x}_i=\hat{A}_ix_i+\sum_{j\in\mathcal{N}_i}\hat{A}_{ij}x_j+F_id_i+B_iu_i    
\end{align}
where the state-space vector, disturbance input and control input are given respectively, by
\begin{align}
    x_i&=[\Delta \delta_i, \Delta \omega_i, \Delta P_{m,i}]^\top\in\mathbb{R}^3, &&  \forall i\in\mathcal{N}\nonumber\\
    d_i&=\Delta P_{L,i}\in\mathbb{R},&& \forall i\in\mathcal{N}\nonumber\\
    u_i&\in\mathbb{R},&& \forall i\in\mathcal{N}.\nonumber
   \end{align}
The matrices $\hat{A}_i$ describing the local dynamics are defined as
\begin{align}
\hat{A}_i=\begin{bmatrix} 0 & 1 & 0 \\
\frac{\omega_b}{M_i}\sum_{j\in\mathcal{N}_i}(-\frac{1}{X_{ij}}) & -\frac{D_i}{M_i} & \frac{\omega_b}{M_i} \\
0 & 0 & -\frac{1}{T_{T,i}}\end{bmatrix}\in\mathbb{R}^{3\times 3}
\end{align}
while the matrices $F_i$ and $B_i$ are defined as
\begin{align}
F_i=\begin{bmatrix} 0 \\
-\frac{\omega_b}{M_i}\\
0\end{bmatrix}\in\mathbb{R}^3,\;\; B_i=\begin{bmatrix} 0 \\
0\\
\frac{1}{T_{T,i}}\end{bmatrix}\in\mathbb{R}^3.
\end{align}
The matrices $\hat{A}_{ij}$ characterizing the interconnections are
\begin{align}
\hat{A}_{ij}=\begin{bmatrix} 0 & 0 & 0 \\
\frac{\omega_b}{M_i}\frac{1}{X_{ij}} & 0 & 0 \\
0 & 0 & 0 \end{bmatrix}\in\mathbb{R}^{3\times 3}.
\end{align}
Without  loss of generality, we let
\begin{align}
d_i=0 
\end{align}
and consider the control input $u_i$ as
\begin{align}
    u_i&=u_{i}^l+u_{i}^g \label{controlinput}\\
    u_{i}^l&=-K_i^\top x_i\\
    u_i^g&=-\sum_{j\in\mathcal{N}_i}K_{ij}^\top x_j
\end{align}
where $u_i^l$ is the local feedback and $u_i^g$ is the global feedback that have to be designed so as to result in guaranteed stability of the decoupled subsystems and minimization of the interconnection terms, respectively.  One can easily verify that each local subsystem is fully controllable with respect to the input $u_i$. By closing the loop using the control input \eqref{controlinput}, we obtain the state-space model:
\begin{align}
\dot{x}_i=A_ix_i+\sum_{j\in\mathcal{N}_i}A_{ij}x_j
\end{align}
where $A_i=(\hat{A}_i-B_iK_i^\top)$ and $A_{ij}=(\hat{A}_{ij}-B_iK_{ij}^\top)$ are the closed-loop system matrices. This model has the form of system \eqref{interconnected_sys}. As mentioned before, when the analysis involves the condition of Theorem~\ref{Theorem_dist_cond_transf}, it may be advantageous to employ the transformed state-space representation. In closed-loop, the transformed state-space model appears as
\begin{align}
\dot{\tilde{x}}_i=\Lambda_i\tilde{x}_i+\sum_{j\in\mathcal{N}_i}\tilde{A}_{ij}\tilde{x}_j
\end{align}
where
\begin{align}
\tilde{B}_i&=T_i^{-1}B_i,\;\;    \tilde{K}_{i}^\top=K_{i}^\top T_i,\;\; \tilde{K}_{ij}^\top=K_{ij}^\top T_j\\
\Lambda_i&=T_i^{-1}\hat{A}_iT_i-\tilde{B}_i\tilde{K}_i^\top\\
\tilde{A}_{ij}&=T_i^{-1}\hat{A}_{ij}T_j-\tilde{B}_i\tilde{K}_{ij}^\top
\end{align}
The local feedback control input $u_{i}^l$ and correspondingly the control gains $\tilde{K}_i$ can be designed using standard pole-placement techniques to ensure that the matrix $A_i$, and correspondingly the matrix $\Lambda_i$, is Hurwitz. On the other hand,  the \textit{optimal} global control input $u_i^g$ that minimizes the effect of interconnections is given by \eqref{optimalglobalcontrol}.

 In our example, we consider the following set of parameters associated with the generators and transmission lines.
\begin{center}
\captionof{table}{Generator parameters.}
\begin{tabular}{ |c|c|c|c|c|c| } 
\hline
Generator &  $T_T$ & $M$ & $D$  \\ 
 \hline
 1  & 0.9 & 8 & 1  \\ \hline
 2  &  1 & 12 & 1  \\ \hline
  3  & 1.1 & 10 & 1 \\ 
 \hline
\end{tabular}
\end{center}
\vspace{1mm}
\begin{center}
\captionof{table}{Line parameters.}
\begin{tabular}{ |c|c| } 
 \hline
 Line & Reactance X  \\ 
 \hline
 1-2 & 0.4  \\ \hline
 1-3 & 0.5  \\ \hline
  2-3 & 0.6  \\
\hline
\end{tabular}
\end{center}
We initiate the DSA algorithm with each agent $i$ acting upon its local control system, adjusting the local control gains $K_{i}$ until stability of its local subsystem is warranted i.e., until $\lambda(A_i)<0$. In our scenario, we assume that the agents independently choose the following sets of eigenvalues
\begin{align}
\lambda_1&=\begin{bmatrix}-22 & -39 & -43  \end{bmatrix} ^\top   \\
\lambda_2&=\begin{bmatrix}-24 & -43 & -37  \end{bmatrix} ^\top   \\
\lambda_3&=\begin{bmatrix}-25 & -38 & -42  \end{bmatrix} ^\top   
\end{align}
This tuning is  performed in parallel by the agents giving rise to the control gains collected as follows
\begin{align}
K_1&=\begin{bmatrix}350.51 & 76.77 & 114.18  \end{bmatrix} ^\top   \\
K_2&=\begin{bmatrix}782.42 & 107.31 & 102.92  \end{bmatrix} ^\top   \\
K_3&=\begin{bmatrix}612.47 & 83.13 & 94.54  \end{bmatrix} ^\top   
\end{align}
Each agent $i$ then computes the transformation matrix $T_i$ that diagonalizes its local closed-loop subsystem matrix $A_i$ and share it with its neighbors $j$, $j\in\mathcal{N}_i$. 
The agents then proceed to compute the interconnection matrices $\tilde{A}_{ij}$ and assess feasibility of Condition~\eqref{condition3}. They carry out this by computing the elements of the matrix $\tilde{S}$ as follows.
\begin{align}
\text{\textbf{Agent 1}}:\; &\tilde{s}_{11}=22,\; \tilde{s}_{12}=296.58,\; \tilde{s}_{13}=249.13\nonumber\\
& |\tilde{s}_{11}|>|\tilde{s}_{12}|+|\tilde{s}_{13}| \Rightarrow \text{"\textit{Condition not met}"}\nonumber\\
\text{\textbf{Agent 2}}:\; &\tilde{s}_{22}=24,\; \tilde{s}_{21}=236.70,\; \tilde{s}_{23}=135.88\nonumber\\
& |\tilde{s}_{22}|>|\tilde{s}_{21}|+|\tilde{s}_{23}| \Rightarrow \text{"\textit{Condition not met}"}\nonumber\\
\text{\textbf{Agent 3}}:\; &\tilde{s}_{33}=25,\; \tilde{s}_{31}=325.07,\; \tilde{s}_{32}=222.14\nonumber\\
& |\tilde{s}_{33}|>|\tilde{s}_{32}|+|\tilde{s}_{31}| \Rightarrow \text{"\textit{Condition not met}"}\nonumber
\end{align}
The agents realize that the distributed stability condition~\eqref{condition3}  of Theorem~\ref{Theorem_dist_cond_transf} cannot be satisfied with reasonably high local control gains and they augment their controls with global control inputs in order to minimize the interconnection terms $\tilde{A}_{ij}$. By deploying the transformation matrices $T_j$ communicated by their neighbors, they compute the optimal control gains using formula \eqref{optimalglobalcontrol} as:
\begin{align}
\tilde{K}_{12}&=\begin{bmatrix}5.33 & -2.91 & -2.30  \end{bmatrix} ^\top   \\
\tilde{K}_{13}&=\begin{bmatrix}4.36 & 2.48 & 2.14  \end{bmatrix} ^\top   \\
\tilde{K}_{21}&=\begin{bmatrix}6.10 & -3.02 & -2.64  \end{bmatrix} ^\top   \\
\tilde{K}_{23}&=\begin{bmatrix} 3.35 &   1.90 &   1.64 \end{bmatrix} ^\top   \\
\tilde{K}_{31}&=\begin{bmatrix} 4.48 & -2.22 & -1.94  \end{bmatrix} ^\top   \\
\tilde{K}_{32}&=\begin{bmatrix} 3.01 & -1.64 &  -1.30  \end{bmatrix} ^\top  
\end{align}
Upon implementing the global inputs $u_i^g$ that minimize $\tilde{A}_{ij}$, they obtain the new elements of the matrix $\tilde{S}$ as follows:
\begin{align}
\text{\textbf{Agent 1}}:\; &\tilde{s}_{11}=22,\; \tilde{s}_{12}=11.15,\; \tilde{s}_{13}=9.37\nonumber\\
& |\tilde{s}_{11}|>|\tilde{s}_{12}|+|\tilde{s}_{13}| \Rightarrow \text{"\textit{Condition met}"}\nonumber\\
\text{\textbf{Agent 2}}:\; &\tilde{s}_{22}=24,\; \tilde{s}_{21}=13,\; \tilde{s}_{23}=7.46\nonumber\\
& |\tilde{s}_{22}|>|\tilde{s}_{21}|+|\tilde{s}_{23}| \Rightarrow \text{"\textit{Condition met}"}\nonumber\\
\text{\textbf{Agent 3}}:\; &\tilde{s}_{33}=25,\; \tilde{s}_{31}=12.23,\; \tilde{s}_{32}=8.36\nonumber\\
& |\tilde{s}_{33}|>|\tilde{s}_{32}|+|\tilde{s}_{31}| \Rightarrow \text{"\textit{Condition met}"}\nonumber
\end{align}
All agents then broadcast to the system operator the message \textit{``Condition met"} to inform him that they verified the distributed stability condition. Subsequently, the system operator invokes Theorem ~\ref{Theorem_dist_cond_transf} and affirms asymptotic stability of the interconnected system $\Sigma_{trans}$, which translates of course into asymptotic stability of the original system $\Sigma$. The eigenvalues of the overall closed-loop system are depicted in Fig.~\ref{fig:eig_Afull}.

\begin{figure}[h!]
    \centering
    \includegraphics[scale=0.42]{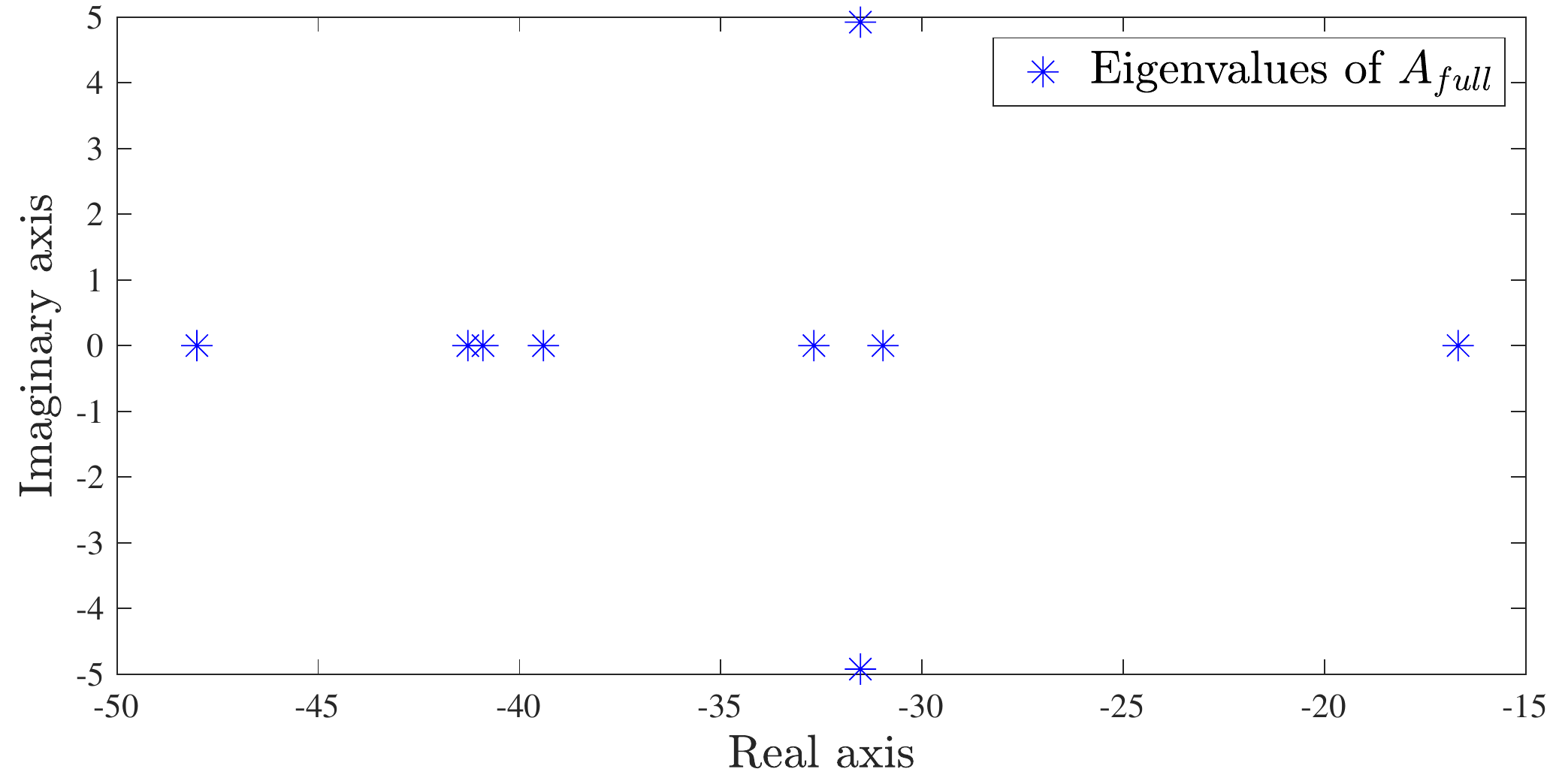}{}
    \caption{Eigenvalues of the overall closed-loop system matrix $A_{full}$.}
    \label{fig:eig_Afull}
    \end{figure}
    
    \begin{figure}[h!]
    \begin{center}
    \begin{subfigure}[b]{0.56\textwidth}
    \includegraphics[scale=0.42]{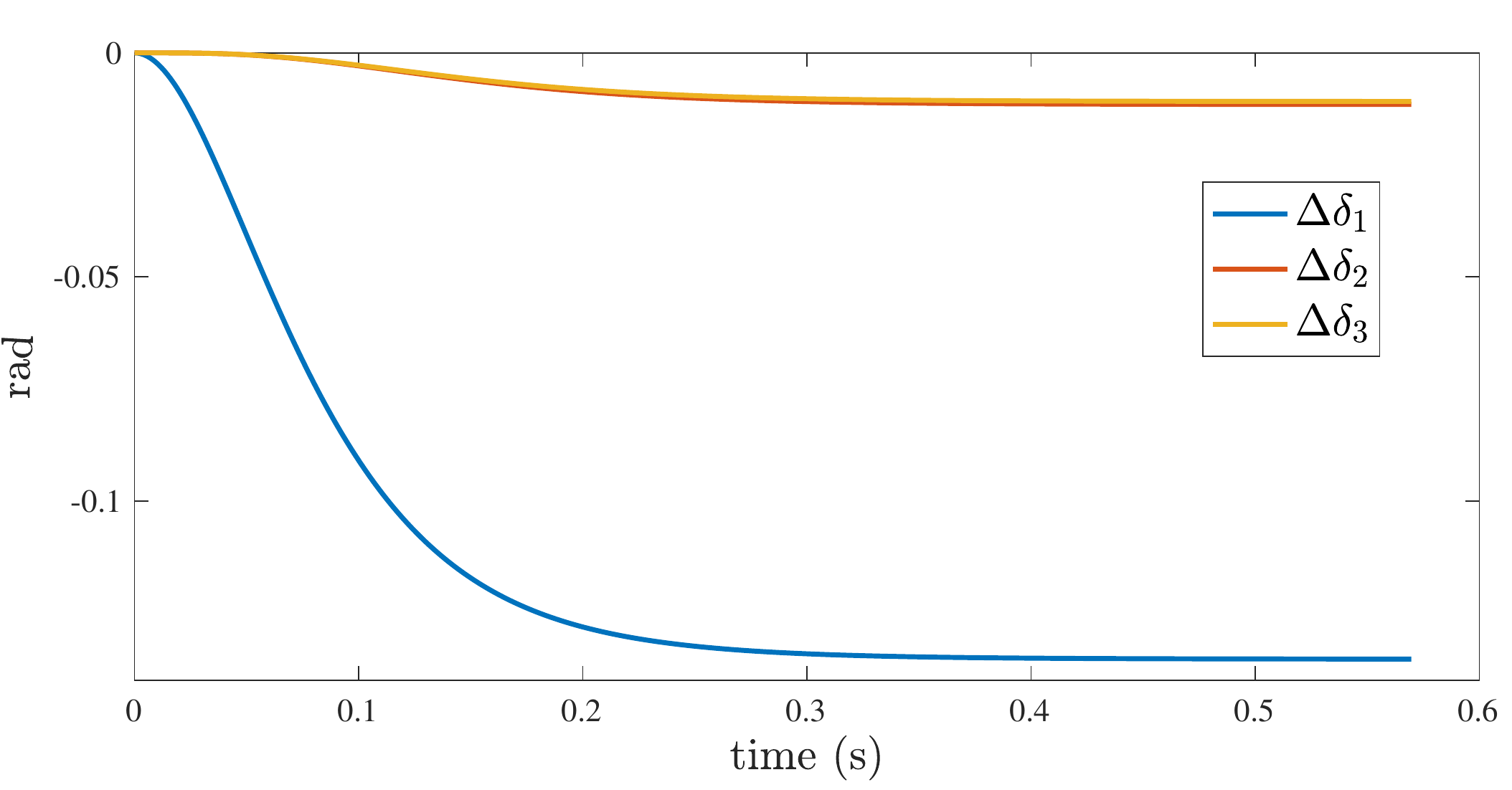}{}
    \caption{Voltage angles $\Delta \delta_i$.}
    \label{fig:Dd}
    \end{subfigure}

    \begin{subfigure}[b]{0.56\textwidth}
    \includegraphics[scale=0.42]{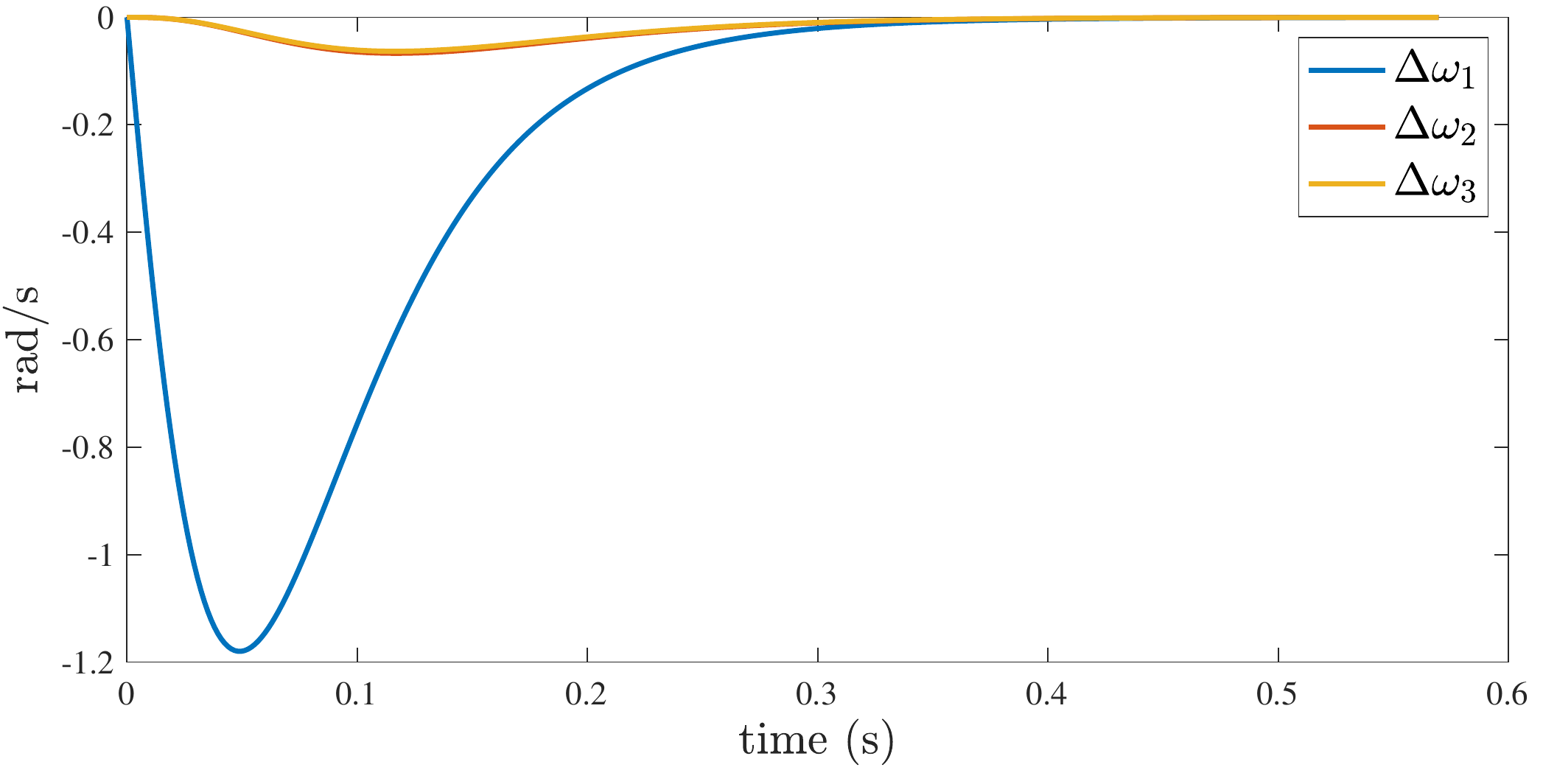}{}
    \caption{Rotor speeds $\Delta\omega_i$.}
    \label{fig:Dw}
    \end{subfigure}

    \begin{subfigure}[b]{0.56\textwidth}
    \includegraphics[scale=0.42]{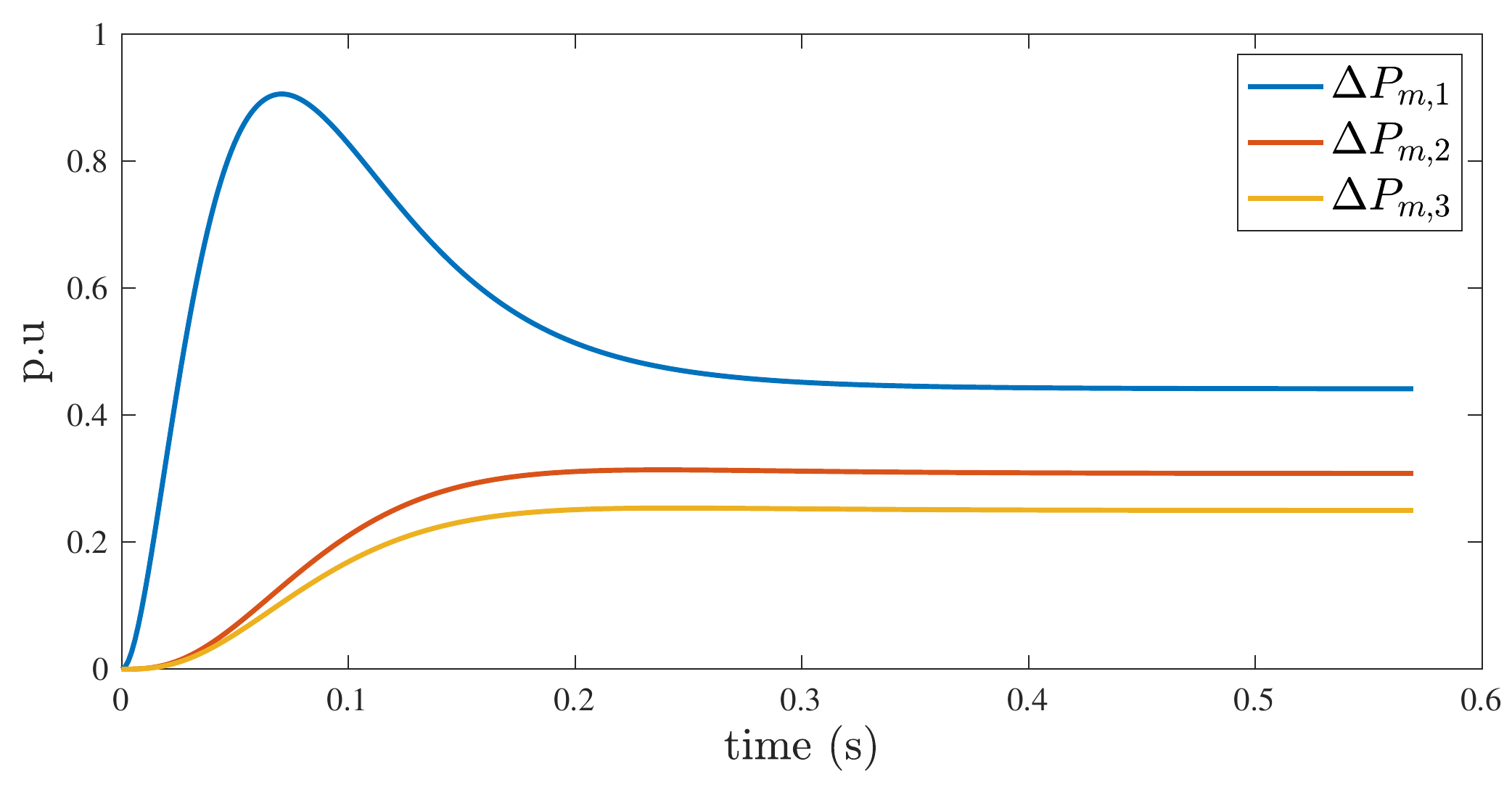}{}
    \caption{Mechanical powers $\Delta P_{m,i}$.}
    \label{fig:DPm}
    \end{subfigure}
    \end{center}
    \caption{System response after a step change in $\Delta P_{L,1}$.}
    \label{sys_response}
        \end{figure}
The system operator finally informs the agents that small-signal stability of the interconnected power grid is guaranteed and that their designed local and global control inputs can be implemented. 
\par To gain some insight on the performance of the interconnected power grid under the locally chosen control gains we simulated a step change in the load of bus 1. The system response is shown in Fig.~\ref{sys_response}. It is clear from this figure that the interconnected system manifests a stable and well-damped dynamic behavior in response to this load disturbance.  It is worthwhile noting that the frequency deviations return back to zero due to the angle feedback. This is shown in Fig.~\ref{fig:Dw}. The larger frequency and angle deviations are observed at bus 1 which is the bus which accommodates the load change. Lastly, from Fig.~\ref{fig:DPm}, we can see how the generators increase their mechanical power outputs to compensate for the load increase at bus 1. 
\par Collectively, the above numerical results corroborate that our proposed DSA methodology can give rise to a small-signal stable interconnected power grid that exhibits good dynamic behavior by only leveraging local stability guarantees. 
\section{Concluding remarks and future work}
In this paper, we present a comprehensive framework for distributed and compositional stability analysis of power grids. Our framework comprises a computationally efficient, privacy preserving and fully distributed methodology for evaluating and certifying stability of power grids. Our methodology mandates that first representative agents at various buses exchange information with their neighbors and design their local controls in order to meet a simple local stability condition. Subsequently, the system operator, by combining the local stability guarantees that are established by the agents, can conclude small-signal stability of the interconnected power grid. We analytically construct the local stability condition and prove that when it is satisfied the interconnected system is guaranteed to be stable. The effectiveness of our proposed distributed stability assessment methodology is illustrated via numerical results centered around a three-bus power grid example. In future work, we would like to explore ways of relaxing the conservativeness of the local stability condition even more.

\end{document}